\documentclass[runningheads]{llncs}
\usepackage[T1]{fontenc}
\usepackage{graphicx}
\usepackage{tikz,pgfplots}
\usepackage{algorithm}
\usepackage{algorithmicx}
\usepackage{algpseudocode}
\usepackage{graphicx}
\usepackage{hyperref}
\usepackage{amsmath}
\usepackage{float}
\usepackage{amssymb}
\usepackage{tabularx}
\usepackage{tikz}
\usepackage{makecell}
\usepackage{array}
\usepackage{enumitem}
\usepackage{comment}
\usepackage{graphicx}
\usepackage{braket}
\usetikzlibrary{positioning,arrows.meta}

\usepackage{float} 
\usepackage{listings}
\floatstyle{ruled}
\newfloat{listing}{htbp}{lop}
\floatname{listing}{Listing}

\begin{document}
\title{Reversible Lifetime Semantics\\for Quantum Programs}
\titlerunning{Reversible Lifetime Semantics for Quantum Programs}

\author{Simone Faro\inst{1}\orcidID{0000-0001-5937-5796} \and
Francesco Pio Marino\inst{1,2}\orcidID{0000-0003-4722-9542} \and
Gabriele Messina\inst{1}\orcidID{0009-0004-9926-3691}}

\institute{Università di Catania, Catania, Italy \and Université de Rouen Normandie, Rouen, France}

\authorrunning{S. Faro, F.P. Marino and G. Messina}

\maketitle              
\begin{abstract}
Reversible computation requires that intermediate data be explicitly undone rather than discarded. In quantum programming, this principle appears as uncomputation, usually treated as a technical cleanup mechanism. We instead present uncomputation as a semantic foundation. In the Qutes language, we introduce a formal model of \emph{Scope-Bounded Liveness-Guided Uncomputation}, where lexical scope bounds variable lifetime and static liveness and entanglement analysis determine the earliest safe reclamation point. We define semantic lifetime and a Restoration Invariant ensuring that temporary quantum information disappears once it becomes semantically irrelevant. We prove compositional correctness under nested scopes and show that early reclamation can reduce circuit depth by avoiding critical-path overhead and can bound peak live qubits through disciplined ancilla reuse. Finally, we show that parameter passing semantics emerges from the same lifetime discipline, with pass-by-value and pass-by-reference corresponding to different lifetime boundaries, and we characterize the constraints (irreversibility, persistent entanglement, and aliasing) under which automatic uncomputation must be restricted. 

\end{abstract}
\section{Introduction}

Reversible computation establishes that information cannot be erased without physical cost, as formalized by Landauer and subsequent foundational work \cite{landauer1961,bennett1973,lecerf1963,toffoli1980}. Quantum circuits are unitary at the gate level, but quantum \emph{programs} routinely introduce temporary data through ancilla allocation, arithmetic subroutines, and structured control flow. Unless explicitly undone, such temporary data may remain entangled with live variables, obstructing interference and compromising algorithmic primitives such as amplitude amplification \cite{grover}. 

Uncomputation restores temporary qubits to their initial state, typically $\ket{0}$, eliminating residual entanglement and preserving reversibility. In low-level circuits this is achieved by explicitly applying inverse segments. In high-level programs, however, manual inversion becomes brittle and non-modular: inverse operations may be scattered across abstractions or interleaved with control flow.

Several approaches address this problem. Silq enforces safe uncomputation through a type discipline \cite{silq}. Unqomp and Reqomp synthesize cleanup circuits and optimize space usage at the compiler level \cite{unqomp,reqomp}. Modular synthesis techniques further support compositional uncomputation \cite{venev2024modular}. Despite their differences, these approaches largely treat uncomputation as a transformation applied after the forward computation has been specified.

We adopt a different perspective. In this work, and in its concrete realization within the high-level quantum programming language \emph{Qutes}\footnote{\emph{Qutes} is a high-level quantum programming language designed to support structured control flow, ownership-aware quantum data, and lifetime-guided uncomputation. Documentation and examples are available at \texttt{https://qutes-lang.org/}.}, uncomputation is governed by \emph{semantic lifetime}. We introduce \emph{Scope-Bounded Liveness-Guided Uncomputation}, a discipline in which lexical scope bounds variable lifetime, while static liveness and entanglement analysis determine the earliest semantically safe reclamation point. Temporary quantum data must satisfy a \emph{Restoration Invariant}: once a variable becomes semantically irrelevant, it must be restored and released.

Within Qutes, lifetime boundaries are first-class structural elements of the language, enabling static reasoning about ownership, aliasing, and entanglement before circuit synthesis is finalized. This model yields structural consequences at both circuit and language levels. Early reclamation is compositional under nested scopes, may avoid unnecessary critical-path accumulation, and bounds peak live width through disciplined ancilla reuse. At the language level, parameter passing semantics emerges from lifetime boundaries: parameters whose lifetime ends at function exit satisfy the Restoration Invariant (observational pass-by-value), whereas explicitly extended lifetimes induce pass-by-reference behavior. We further characterize the precise constraints under which automatic reclamation is restricted, namely irreversibility, persistent entanglement, and aliasing.

Throughout the paper we assume familiarity with the standard circuit model of quantum computation. We consider unit-cost one- and two-qubit gates and measure complexity in terms of circuit depth and size \cite{NielsenChuang}. Basic notions of quantum computation are not reviewed for space reasons. Within this setting, we show that lifetime-based uncomputation, as realized in Qutes, is not a post-processing optimization, but a compositional semantic principle that unifies reversible reasoning, resource management, and high-level quantum programming.

\section{Motivation, Related Work, and Overview of Qutes}

Ancillary qubits are unavoidable in non-trivial quantum programs. They arise in arithmetic routines, structured control flow, and intermediate data storage. If not properly restored, ancillas may remain entangled with live outputs, thereby altering observable behavior and preventing interference-based primitives such as amplitude amplification \cite{grover}. The management of temporary quantum data is therefore not merely an implementation concern, but a semantic requirement of reversible computation.

In low-level circuit descriptions, uncomputation is achieved by explicitly appending inverse segments. At higher abstraction levels, however, inverse operations may be distributed across control structures and modular boundaries. This has motivated language- and compiler-level support for automatic cleanup.

Several approaches exist. Silq enforces safe uncomputation via a type system that guarantees restoration of temporary values \cite{silq}. Q\# and Quipper provide explicit adjoint generation mechanisms and structured compute--use--uncompute patterns \cite{qsharp,Quipper}. Compiler-based techniques such as Unqomp and Reqomp synthesize inverse circuits and optimize space usage after the forward circuit has been constructed \cite{unqomp,reqomp}. More recent modular synthesis approaches support compositional cleanup at the circuit level \cite{venev2024modular}.

Despite their differences, these systems largely treat uncomputation as:
(i) a correctness device, and
(ii) a transformation applied to a computation whose dataflow has already been fixed.

In contrast, Qutes is designed so that variable lifetime and ownership are explicit at the language level \cite{FaroMM25,Qutes_Journal}. It is a high-level quantum programming language based on structured control flow, lexical scopes, and array-oriented quantum data. Temporary registers are introduced within syntactic regions whose boundaries delimit their intended lifetime. This structural information enables static reasoning about liveness and entanglement before circuit synthesis is finalized.

The key distinction is therefore semantic rather than purely operational: in Qutes, uncomputation is governed by lifetime analysis derived from scope and liveness, rather than by explicit adjoint blocks or by post-hoc circuit rewriting.

\medskip

\begin{table}[t]
\centering
\caption{Comparison of uncomputation mechanisms across languages and frameworks.}
\label{tab:comparison}
{\scriptsize
\begin{tabular}{lllll}
\hline
\textbf{System} 
& \textbf{Automatic} 
& \textbf{Scope-aware} 
& \textbf{Lifetime-based} 
& \textbf{Early Reclaim} \\
\hline
&&&&\\[-0.2cm]
Silq \cite{silq} 
& Type-enforced 
& Partial 
& No explicit lifetime model 
& No \\[0.1cm]
Q\# \cite{qsharp} 
& Adjoint blocks 
& No 
& No 
& No \\[0.1cm]
Quipper \cite{Quipper} 
& Adjoint generation 
& No 
& No 
& No \\[0.1cm]
Unqomp \cite{unqomp} 
& Compiler synthesis 
& No 
& No 
& Limited \\[0.1cm]
Reqomp \cite{reqomp} 
& Space-aware synthesis 
& No 
& No 
& Limited \\
&&&&\\[-0.2cm]
\hline
&&&&\\[-0.2cm]
Qutes (this work) 
& Yes 
& Yes 
& Yes 
& Yes \\
&&&&\\[-0.2cm]
\hline
\end{tabular}
}
\end{table}

Table~\ref{tab:comparison} summarizes the main differences between representative systems and Qutes with respect to uncomputation support.
Existing systems either rely on explicit adjoint constructs or synthesize cleanup segments once the forward circuit has been constructed. Qutes differs in that lifetime boundaries are first-class structural elements of the language. This enables liveness-guided reclamation within nested scopes, potentially reducing both critical-path accumulation and peak live width.

Beyond quantum-specific settings, the proposed lifetime discipline may be viewed as a resource-aware refinement of classical reversible computation. While Bennett-style cleanup restores global state after a forward phase, lifetime-guided reclamation enforces local restoration at semantic boundaries. In this sense, the approach can be interpreted as introducing a structured, region-like discipline for reversible resource management, aligned with the foundational principles of logical reversibility but tailored to quantum entanglement and ownership constraints.
The remainder of the paper formalizes this lifetime-based discipline and studies its circuit-level and semantic consequences.


\section{Automatic Uncomputation in Qutes}

This section formalizes \emph{Scope-Bounded Liveness-Guided Uncomputation}. 
We introduce a dependence-based program model and define liveness and semantic lifetime for quantum variables. These notions determine when a temporary register can be safely restored and released, and allow us to analyze the impact of early reclamation on circuit depth and width. We further study compositionality under nested scopes.

\smallskip

As a running example (Fig.~\ref{fig:running-example}), consider a function 
\texttt{compute} with inputs $x_1,x_2$ and outputs $y_1,y_2$, using temporaries $t_1,\dots,t_4$. 
At the circuit level, the abstract Qutes procedures $f,g,h,k$ correspond to unitary blocks $U_f,U_g,U_h,U_k$. 
The first segment prepares $(t_1,t_2)$ from $x_1$ and transfers their contribution to $y_1$; the second segment prepares $(t_3,t_4)$ from $(x_1,x_2)$ and transfers their contribution to $y_2$, followed by $k$ on $y_2$. 
Assuming outputs initialized to $\ket{0}$, the function computes
\[
y_1 = f(x_1) \oplus g(x_1),
\qquad
y_2 = k\big(h(x_1)\oplus h(x_2)\big).
\]

The example also illustrates the classical compute--then--uncompute strategy, where adjoint blocks are appended after the forward computation. In contrast, our analysis will show how reclamation can be driven by semantic lifetime rather than by global circuit termination. The example is intentionally minimal yet captures dependence chains, entanglement propagation, and temporaries whose last syntactic use differs from their last semantic influence.

\begin{figure}[!t]
\centering
{\footnotesize
\begin{tabular}{ll}
\hline
\multicolumn{2}{l}{\texttt{qubit[] compute(qubit $x_1$, qubit $x_2$, qubit $y_1$, qubit $y_2$) \{}}\\
$p_0$. & \quad \texttt{qubit $t_1$, $t_2$, $t_3$, $t_4$;}\\
$p_1$. & \quad \texttt{f($x_1$,$t_1$);}\\
$p_2$. & \quad \texttt{g($x_1$,$t_2$);}\\
$p_3$. & \quad \texttt{CNOT $t_1$,$t_2$;}\\
$p_4$. & \quad \texttt{CNOT $t_2$,$y_1$;}\\
$p_5$. & \quad \texttt{h($x_1$,$t_3$);}\\
$p_6$. & \quad \texttt{h($x_2$,$t_4$);}\\
$p_7$. & \quad \texttt{CNOT $t_3$,$y_2$;}\\
$p_8$. & \quad \texttt{CNOT $t_4$,$y_2$;}\\
$p_9$. & \quad \texttt{k($y_2$);}\\
$p_{10}$. & \quad \texttt{return [$y_1$, $y_2$];}\\
 & \texttt{\}}\\
\end{tabular}
}
\vspace{0.5em}
\includegraphics[width=0.85\textwidth]{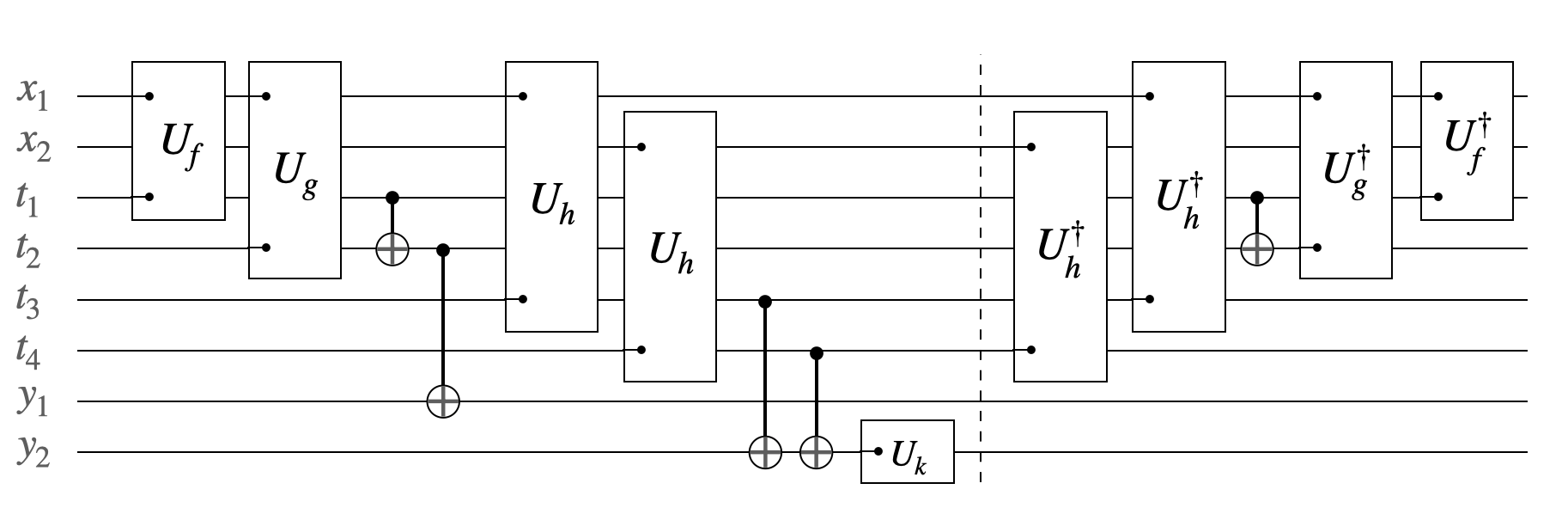}
\caption{Circuit-level realization of the running example. The forward computation (left) implements the unitary blocks $U_f$, $U_g$, $U_h$, and $U_k$, while the rightmost segment illustrates the classical compute--then--uncompute strategy in which adjoint blocks are appended at the end of the circuit.}\label{fig:running-example}
\end{figure}


\subsection{Formal Model of Lifetime-Guided Uncomputation}

We formalize liveness and semantic lifetime in terms of dependence and entanglement relations between unitary segments.

Let a compiled Qutes program be represented as a sequence of unitary blocks
\[
C = (U_1, U_2, \dots, U_n),
\]
where each $U_i$ acts on a subset of quantum variables.
Program points are indexed by $1 \le i \le n$, with point $i$ denoting the state immediately after $U_i$.

Our formal development abstracts away from surface syntax and defines lifetime semantics directly at the level of compiled unitary sequences. This abstraction should be understood as a semantic contract between the high-level language and its circuit-level realization: any Qutes program is assumed to admit a unitary decomposition consistent with lexical scope and ownership constraints. The results established below therefore characterize the semantic discipline that the compiler must enforce, independently of concrete syntactic constructs.

\begin{definition}[Data Dependence Graph]
Let $C = (U_1,\dots,U_n)$. The \emph{data dependence graph}
$G_D = (V,E_D)$ is defined as follows:
\begin{itemize}
\item $V = \{1,\dots,n\}$ consists of program points.
\item $(i \to j) \in E_D$ with $i<j$ if $U_j$ depends on $U_i$
      through a shared quantum variable.
\end{itemize}
In particular, if a variable is written by $U_i$ and later read or
modified by $U_j$, then $(i \to j) \in E_D$.
\end{definition}

The graph $G_D$ captures functional data flow and induced ordering constraints.
Under optimal scheduling, the circuit depth equals the length of the longest directed path in $G_D$.

In the running example (Fig.~\ref{fig:running-example}), the corresponding dependence graph
(Fig.~\ref{fig:entanglement-graph-example} on the left) makes explicit how temporary registers propagate
their influence to the outputs through controlled interactions. 
In particular, dependencies induced by multi-qubit gates extend the chain of influence from temporaries to live outputs.

\medskip

\begin{definition}[Entanglement Graph]
At program point $p$, the \emph{entanglement graph}
$G_E(p) = (Q, E_E(p))$ is an undirected graph where:
\begin{itemize}
\item $Q$ is the set of quantum variables,
\item $(q_i,q_j) \in E_E(p)$ if previous unitary operations
      may have induced non-separable correlations
      between $q_i$ and $q_j$.
\end{itemize}
For compilation, $G_E(p)$ is conservatively approximated by
tracking multi-qubit gates: whenever a gate jointly acts on two variables,
an undirected edge is introduced between them.
\end{definition}

In Fig.~\ref{fig:entanglement-graph-example} (on the right), the entanglement graph of the running example
illustrates how controlled operations introduce correlations that may extend
the semantic influence of temporary variables beyond their last syntactic use.
The entanglement relation will be used to conservatively propagate liveness.

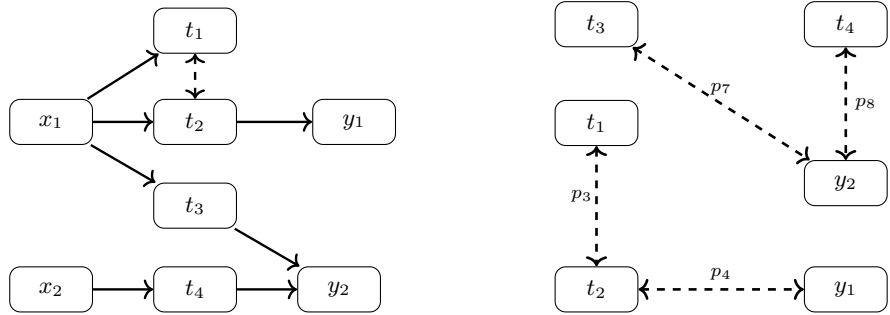
\begin{figure}[t]
\centering
\begin{tikzpicture}[
  node distance=1.6cm and 1.8cm,
  var/.style={draw, rounded corners, minimum width=1.1cm, minimum height=0.6cm},
  data/.style={->, line width=0.9pt},
  ent/.style={<->, dashed, line width=0.9pt},
  font=\small
]

\node[var] (x1) {$x_1$};
\node[var, right=0.8cm of x1] (t2) {$t_2$};
\node[var, above=6mm of t2] (t1) {$t_1$};
\node[var, below=5mm of t2] (t3) {$t_3$};
\node[var, below=5mm of t3] (t4) {$t_4$};

\node[var, left=0.8cm of t4] (x2) {$x_2$};
\node[var, right=1cm of t2] (y1) {$y_1$};
\node[var, right=0.8cm of t4] (y2) {$y_2$};

\draw[data] (x1) -- (t1);
\draw[data] (x1) -- (t2);
\draw[data] (x1) -- (t3);
\draw[data] (x2) -- (t4);

\draw[data] (t2) -- (y1);
\draw[data] (t3) -- (y2);
\draw[data] (t4) -- (y2);

\draw[ent] (t1) -- (t2);


\end{tikzpicture}\hspace{2cm}
\begin{tikzpicture}[
  node distance=1.6cm and 2.2cm,
  var/.style={draw, rounded corners, minimum width=1.1cm, minimum height=0.6cm},
  ent/.style={<->, dashed, line width=1.0pt},
  font=\small
]

\node[var] (t1) {$t_1$};
\node[var, below=of t1] (t2) {$t_2$};
\node[var, right=of t2] (y1) {$y_1$};

\node[var, above=0.8cm of y1] (y2) {$y_2$};
\node[var, above=1.5cm of y2] (t4) {$t_4$};
\node[var, left=of t4] (t3) {$t_3$};

\draw[ent] (t1) -- node[above left=-2pt and -2pt] {\scriptsize $p_3$} (t2);
\draw[ent] (t2) -- node[above] {\scriptsize $p_4$} (y1);

\draw[ent] (t3) -- node[above] {\scriptsize $p_7$} (y2);
\draw[ent] (t4) -- node[right] {\scriptsize $p_8$} (y2);

\end{tikzpicture}
\caption{(on the left) Data Dependence Graph for the example of Fig. \ref{fig:running-example}. Solid arrows denote functional data dependencies, while dashed bidirectional edges represent entanglement-induced dependencies.(on the right) Entanglement Graph for the example: dashed edges connect variables that may become entangled through controlled operations.}
\label{fig:entanglement-graph-example}
\end{figure}


\begin{definition}[Liveness]
A quantum variable $v$ is \emph{live} at program point $p$
if there exists a future operation $U_j$, with $j > p$, such that:
\begin{enumerate}
\item $v$ participates directly in $U_j$, or
\item there exists a variable $w$ such that
      $(v,w) \in E_E(p)$ and $w$ is live at $p$.
\end{enumerate}
We denote this predicate by $\mathrm{live}(v,p)$.
\end{definition}

Condition (2) propagates liveness through entanglement, ensuring that
a variable remains live as long as it is correlated with live data.

In the example of Fig. \ref{fig:running-example}, the liveness relation propagates through both data dependencies and entanglement edges. After $p_3$, the gate $\mathrm{CX}(t_1,t_2)$ makes $t_1$ and $t_2$ mutually dependent, and the subsequent gate $\mathrm{CX}(t_2,y_1)$ at $p_4$ connects this pair to the live output $y_1$. Since $y_1$ remains live until the end of the computation, liveness conservatively propagates back to both $t_2$ and $t_1$. Similarly, $t_3$ and $t_4$ become live at $p_5$ and $p_6$, and remain live after their interaction with $y_2$ at $p_7$ and $p_8$, as $y_2$ is used at $p_9$. This illustrates how entanglement-induced dependencies extend the live range of temporary variables beyond their last syntactic use.


\begin{definition}[Semantic Lifetime]
Let $v$ be a quantum variable declared within a lexical scope $S$.
The \emph{semantic lifetime} of $v$ is the maximal interval
\[
[p_{\min}, p_{\max}]
\]
such that $\mathrm{live}(v,p)$ holds for every program point
$p \in [p_{\min}, p_{\max}]$, where
\[
p_{\min} = \min \{\, p \mid \mathrm{live}(v,p) \,\}, 
\qquad
p_{\max} = \max \{\, p \mid \mathrm{live}(v,p) \,\}.
\]

The lexical scope of $v$ provides a syntactic upper bound on
$p_{\max}$ and a lower bound on possible allocation, but the semantic
lifetime may begin strictly after declaration and terminate strictly
before the end of the scope.

This lifetime is computed with respect to the entanglement-based
liveness relation and therefore constitutes a conservative
approximation of semantic relevance.
\end{definition}


Observe that, throughout the formal development, liveness and semantic lifetime are defined with respect to quantum resources rather than syntactic identifiers. If multiple program variables alias the same underlying quantum register, they are treated as a single resource for the purposes of liveness analysis. Formally, if $v_1, \dots, v_k$ reference the same physical register $R$, then
\[
\mathrm{live}(R,p) \;\text{holds iff}\; \exists i \; \mathrm{live}(v_i,p),
\]
and the semantic lifetime of $R$ is the union of the lifetimes of all its aliases. Safe reclamation is therefore permitted only when the resource itself is no longer live.


\begin{definition}[Output-Isolable Subcircuit]
Let $U$ be the unitary transformation corresponding to the
subcomputation between program points $p_a$ and $p_b$.
Let $L$ be the set of live output variables at $p_b$.
The subcircuit is said to be \emph{output-isolable} with respect to a
set of temporary variables $T$ if there exists a decomposition
\[
U = U_L \circ V_T
\]
such that:
\begin{itemize}
\item $U_L$ acts only on variables in $L$,
\item $V_T$ acts only on variables in $T$,
\item and $V_T$ is invertible without altering the reduced state of $L$.
\end{itemize}
\end{definition}

Output-isolability is a semantic condition rather than a purely syntactic one. In general, determining whether a unitary admits such a decomposition may require global reasoning about entanglement structure. In practice, the static analysis implemented in Qutes provides a conservative approximation: whenever isolability cannot be certified, automatic reclamation is disabled. The formal results below therefore assume isolability as a semantic property, while acknowledging that its static detection is necessarily incomplete.

\paragraph{Remark.}
Entanglement with live variables implies liveness and therefore
extends the semantic lifetime conservatively.
However, entanglement alone does not necessarily prevent safe
uncomputation.
If a subcircuit is output-isolable, the temporary variables involved
may admit a valid early reclamation even if they are entangled
according to the static liveness analysis.
The entanglement-based lifetime should therefore be interpreted as a
sound but not complete approximation of semantic irreversibility.

In the example of Fig. \ref{fig:running-example}, the entanglement-based liveness analysis yields a conservative semantic lifetime in which $t_1$ becomes live at $p_1$ and, after the interaction $\mathrm{CX}(t_1,t_2)$ at $p_3$ and $\mathrm{CX}(t_2,y_1)$ at $p_4$, remains live through the end of the computation due to entanglement propagation with the live output $y_1$. Similarly, $t_2$ becomes live at $p_2$ and is conservatively live until the end. Variables $t_3$ and $t_4$ become live at $p_5$ and $p_6$, respectively, and remain live until the end because they are entangled with the live output $y_2$ via $p_7$ and $p_8$, and $y_2$ is later consumed at $p_9$.

However, after $p_4$ the subcomputation that prepares $(t_1,t_2)$ is output-isolable: the only operation involving the output $y_1$ is $\mathrm{CX}(t_2,y_1)$, while the preceding history on $(t_1,t_2)$ consists of $p_1$--$p_3$. Therefore, immediately after $p_4$ the compiler may insert the adjoint sequence $\mathrm{CX}(t_1,t_2)$, $g^\dagger(x_1,t_2)$, $f^\dagger(x_1,t_1)$, which restores $t_1$ and $t_2$ to their initial states without affecting the final value of $y_1$. As a consequence, under early uncomputation the effective semantic lifetimes of $t_1$ and $t_2$ shrink to $[p_1,p_4]$ and $[p_2,p_4]$, respectively, while the lifetimes of $t_3$ and $t_4$ remain $[p_5,p_9]$ and $[p_6,p_9]$.

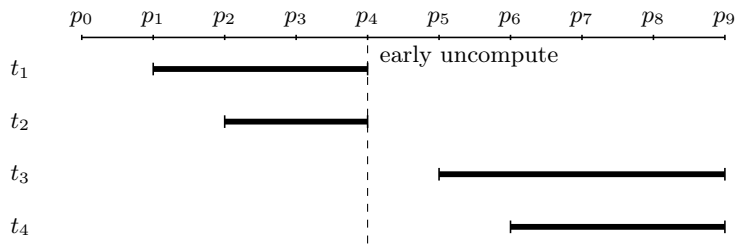
\begin{figure}[t]
\centering
\begin{tikzpicture}[
  x=0.95cm, y=0.7cm,
  font=\small,
  lifeline/.style={line width=2.2pt},
  tick/.style={line width=0.6pt},
  label/.style={anchor=east},
]
  \def\xmin{0}
  \def\xmax{9}

  \def\ytone{0}
  \def\yttwo{-1}
  \def\ytthree{-2}
  \def\ytfour{-3}

  \draw[tick] (\xmin,0.6) -- (\xmax,0.6);
  \foreach \i in {0,...,9} {
    \draw[tick] (\i,0.55) -- (\i,0.65);
    \node[anchor=south] at (\i,0.65) {$p_{\i}$};
  }

  \node[label] at (-0.6,\ytone) {$t_1$};
  \node[label] at (-0.6,\yttwo) {$t_2$};
  \node[label] at (-0.6,\ytthree) {$t_3$};
  \node[label] at (-0.6,\ytfour) {$t_4$};

  \draw[lifeline] (1,\ytone) -- (4,\ytone);
  \draw[tick] (1,\ytone-0.12) -- (1,\ytone+0.12);
  \draw[tick] (4,\ytone-0.12) -- (4,\ytone+0.12);

  \draw[lifeline] (2,\yttwo) -- (4,\yttwo);
  \draw[tick] (2,\yttwo-0.12) -- (2,\yttwo+0.12);
  \draw[tick] (4,\yttwo-0.12) -- (4,\yttwo+0.12);

  \draw[lifeline] (5,\ytthree) -- (9,\ytthree);
  \draw[tick] (5,\ytthree-0.12) -- (5,\ytthree+0.12);
  \draw[tick] (9,\ytthree-0.12) -- (9,\ytthree+0.12);

  \draw[lifeline] (6,\ytfour) -- (9,\ytfour);
  \draw[tick] (6,\ytfour-0.12) -- (6,\ytfour+0.12);
  \draw[tick] (9,\ytfour-0.12) -- (9,\ytfour+0.12);

  \draw[dashed] (4,0.45) -- (4,-3.45);
  \node[anchor=west] at (4.05,0.25) {\footnotesize early uncompute};
\end{tikzpicture}
\caption{Semantic lifetimes under early uncomputation of the example in Fig. \ref{fig:running-example}. While entanglement-based liveness conservatively extends $t_1$ and $t_2$ to the end of the program, the subcircuit preceding $p_4$ is output-isolable; hence $t_1$ and $t_2$ can be safely reclaimed after $p_4$, whereas $t_3$ and $t_4$ remain live until $p_9$ due to their interaction with $y_2$.}\label{fig:semantic-lifetime-early-uncompute}
\end{figure}

\begin{definition}[Safe Reclamation Point]
A program point $p$ is a \emph{safe reclamation point} for a variable $v$ if $v$ is not live at $p$ and no irreversible operation applied to $v$ has effects that persist beyond $p$. At such a point, the compiler may insert the adjoint of the unitary history applied to $v$, restoring it to its initial state.
\end{definition}

In the running example (Fig.~\ref{fig:running-example}), entanglement-based liveness conservatively extends the lifetime of $t_1$ and $t_2$ beyond their last interaction with $y_1$. However, the segment preparing $(t_1,t_2)$ is output-isolable at that boundary, so the adjoint sequence may be inserted immediately after their final semantic influence on $y_1$. By contrast, the temporaries $(t_3,t_4)$ cannot be reclaimed before their contribution to $y_2$ is fully resolved, since restoring them earlier would alter the reduced state of a live output.

\begin{definition}[Critical Path]
Given the data dependence graph $G_D$, the \emph{critical path length} is the length of the longest directed path in $G_D$. Under optimal scheduling, this equals the circuit depth.
\end{definition}

In Fig.~\ref{fig:entanglement-graph-example}, the critical path corresponds to the longest dependence chain leading to the outputs and reflects the maximal propagation of data and entanglement-induced dependencies, independently of whether temporary variables may later admit early reclamation.

\subsection{Depth Preservation and Reduction}

A central consequence of Scope-Bounded Liveness-Guided Uncomputation concerns circuit depth. In a classical compute--then--uncompute strategy, a temporary unitary segment $U_f$ is executed forward and its adjoint $U_f^\dagger$ is appended only after the entire forward computation. Since both segments lie on the critical path, the depth contribution of $U_f$ is effectively doubled.

Under lifetime-based reclamation, $U_f^\dagger$ is inserted at the earliest safe reclamation point, once the associated variable becomes non-live. This permits local scheduling of $U_f^\dagger$ and allows overlap with subsequent forward computation whenever permitted by the data dependence graph $G_D$. Depth is therefore determined by interaction with the critical path rather than by global syntactic ordering of adjoint blocks.

In the running example (Fig.~\ref{fig:running-example}), the preparation of $(t_1,t_2)$ would, under global cleanup, be inverted only after the full computation completes. Under early reclamation, the adjoint segment is inserted immediately after its last semantic influence on $y_1$, as shown in Fig.~\ref{fig:depth-reduce}. In the modified circuit, the inverse blocks are no longer appended as a trailing suffix but are placed locally at the lifetime boundary, where they may overlap with subsequent computation on disjoint qubits. The resulting schedule eliminates the artificial duplication of the temporary segment along the critical path.

\begin{figure}[!t]
    \centering
    \includegraphics[width=0.85\linewidth]{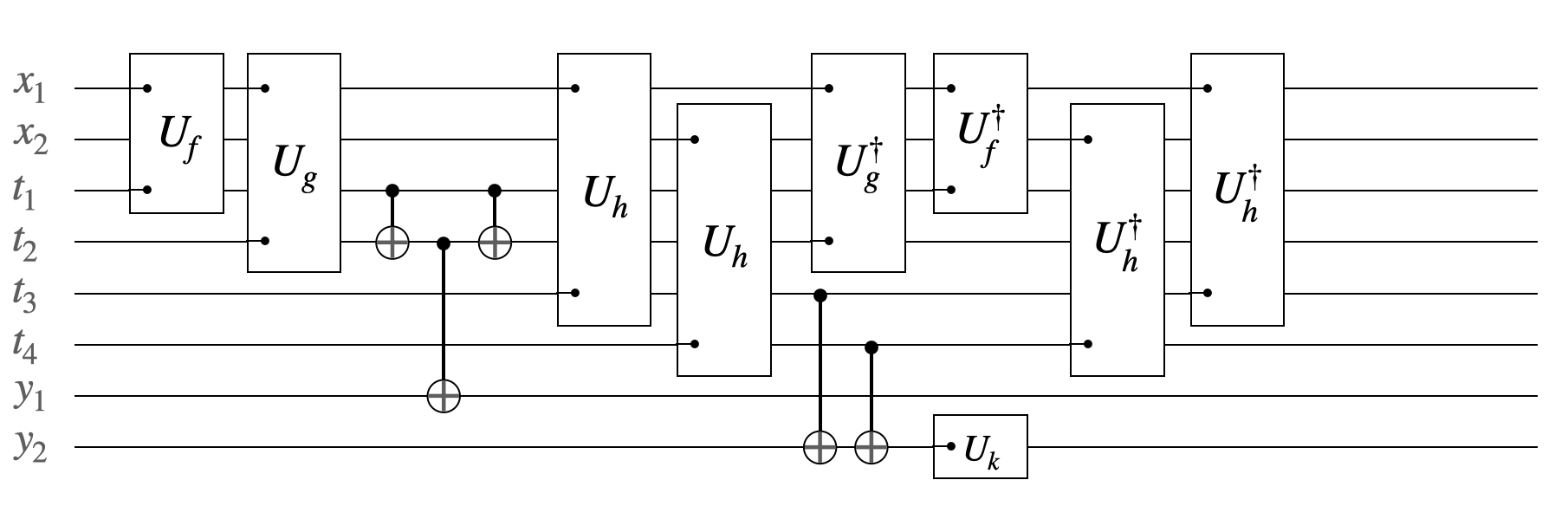}
    \caption{Circuit with early uncomputation: the adjoint blocks associated with $(t_1,t_2)$ are inserted immediately after their last semantic use, allowing overlap with the subsequent computation and reducing the overall circuit depth compared to global compute--then--uncompute.}
    \label{fig:depth-reduce}
\end{figure}

\begin{lemma}[Depth Bound under Early Reclamation]
Let $v$ be a temporary variable associated with a unitary segment $U_f$, and let $U_f^\dagger$ be inserted at its earliest safe reclamation point. Let $U_{\mathrm{rest}}$ denote the remaining forward computation. Under optimal scheduling consistent with $G_D$,
\[
\mathrm{depth}(C)
\;\le\;
\mathrm{depth}(U_f)
+
\max\!\big(
\mathrm{depth}(U_f^\dagger),
\mathrm{depth}(U_{\mathrm{rest}})
\big).
\]
\end{lemma}

\begin{proof}
Once $v$ is non-live, $U_f^\dagger$ depends only on the qubits involved in $U_f$. Its contribution to depth is governed by its interaction with the critical path of $U_{\mathrm{rest}}$. If supports are disjoint, the segments execute in parallel; otherwise, accumulation follows the longest dependency chain in $G_D$. \qed
\end{proof}

Thus early reclamation never increases asymptotic depth relative to global compute--then--uncompute and may strictly reduce it when $U_f^\dagger$ does not lie on the critical path of $U_{\mathrm{rest}}$.

\begin{lemma}[Zero-Overhead Condition]
If (i) $v$ is not live beyond point $p$, (ii) $U_f^\dagger$ acts on qubits disjoint from the critical path of $U_{\mathrm{rest}}$, and (iii) $\mathrm{depth}(U_{\mathrm{rest}}) \ge \mathrm{depth}(U_f^\dagger)$, then inserting $U_f^\dagger$ at $p$ does not increase circuit depth.
\end{lemma}

\begin{proof}
Under disjoint support, $U_f^\dagger$ and $U_{\mathrm{rest}}$ can be scheduled concurrently. Since $U_{\mathrm{rest}}$ determines the critical path, total depth remains unchanged. \qed
\end{proof}

In the ideal case, lifetime-guided uncomputation incurs zero additional depth, unlike global cleanup strategies where adjoint segments necessarily extend the critical path.

There exist structured programs in which a temporary segment of depth $d$
lies on the critical path under a global compute--then--uncompute strategy,
yielding total depth $2d + O(1)$, whereas lifetime-guided insertion of the
adjoint segment allows overlap with independent forward computation, yielding
depth $d + O(1)$. Thus lifetime-based reclamation may achieve strict asymptotic improvement over global cleanup in the presence of independent subcomputations.


\subsection{Ancilla Reuse and Width Reduction}

Beyond circuit depth, lifetime-guided uncomputation directly impacts circuit width, i.e., the number of simultaneously live quantum variables. In global compute--then--uncompute strategies, temporary ancillas often remain allocated until the end of the computation, increasing peak resource usage. Under early reclamation, temporaries are restored and released as soon as their semantic lifetime terminates, allowing physical qubits to be reused.

\begin{definition}[Live Set]
At program point $p$, the \emph{live set} is
\[
L(p) = \{\, v \mid \mathrm{live}(v,p) \,\},
\qquad
W_{\max} = \max_p |L(p)|.
\]
\end{definition}

The quantity $W_{\max}$ bounds the number of qubits required under optimal reclamation.

In the running example (Fig.~\ref{fig:running-example}), global cleanup would keep $t_1$ and $t_2$ allocated while $(t_3,t_4)$ are active, increasing peak width. Under early reclamation (Fig.~\ref{fig:width-reduced}), the adjoint blocks associated with $(t_1,t_2)$ are inserted at their lifetime boundary, so these temporaries are restored before the subsequent subcomputation reaches peak activity. In the modified circuit, the allocation intervals of $(t_1,t_2)$ and $(t_3,t_4)$ no longer overlap, directly reducing $W_{\max}$.

\begin{lemma}[Width Bound under Early Reclamation]
Let $v$ be a temporary variable whose semantic lifetime ends at program point $p$. Then for all $q>p$, $v\notin L(q)$. Consequently, inserting $U_f^\dagger$ at $p$ reduces $|L(q)|$ for all subsequent $q$ in which $v$ would otherwise remain live.
\end{lemma}

\begin{proof}
By definition, $\mathrm{live}(v,p)$ does not hold for $q>p$. The Restoration Invariant ensures that $v$ is returned to its initial state at or before $p$, so it need not remain allocated thereafter. \qed
\end{proof}

\begin{lemma}[Peak Width Reduction]
Let $v_1,\dots,v_k$ be temporary variables with pairwise disjoint semantic lifetimes. Then there exists a schedule under which all $v_i$ share the same physical qubit, and $W_{\max}$ does not increase with $k$.
\end{lemma}

\begin{proof}
At any program point, at most one $v_i$ belongs to $L(p)$. Since each $v_i$ is restored at the end of its lifetime, the same physical qubit may be reused sequentially. \qed
\end{proof}

Thus circuit width is governed by semantic lifetime rather than lexical allocation. Lifetime-guided uncomputation induces reversible resource pooling: temporary qubits are reused once they cease to be live. In the worst case, width matches that of global cleanup; in structured programs with short and non-overlapping lifetimes, peak width is strictly reduced.

\begin{figure}[!t]
    \centering
    \includegraphics[width=0.85\linewidth]{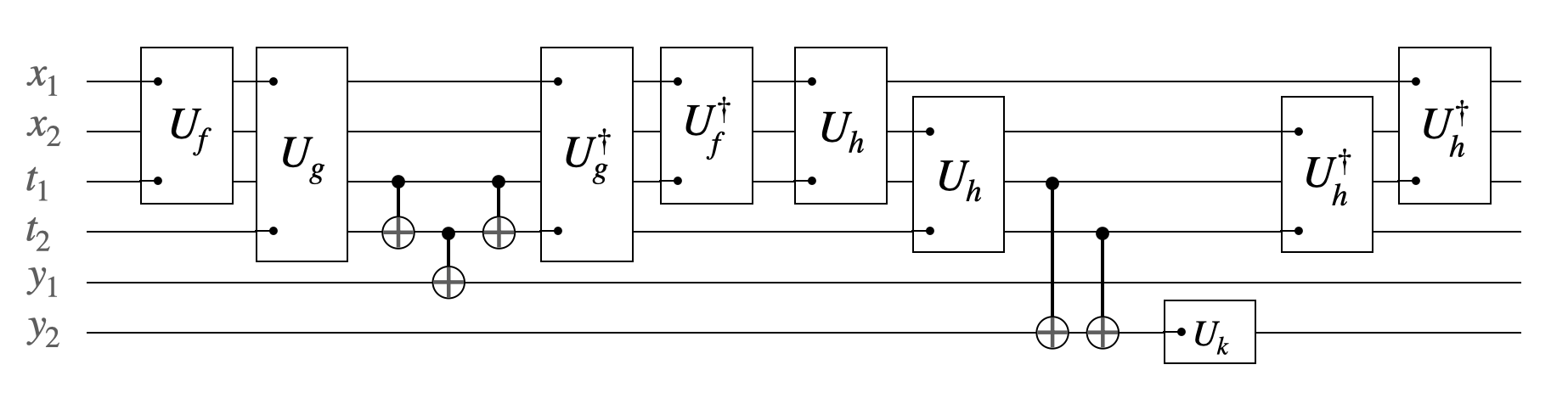}
    \caption{Width-reduced circuit obtained via early reclamation: the temporaries $t_1$ and $t_2$ are restored immediately after their last semantic use, preventing overlap with $t_3$ and $t_4$ and reducing the peak number of simultaneously live qubits.}
    \label{fig:width-reduced}
\end{figure}


\subsection{Compositionality of Early Reclamation}

A central property of Scope-Bounded Liveness-Guided Uncomputation is its compositional behavior under nested lexical scopes. Since reclamation is triggered by semantic lifetime rather than global structure, early reclamation within an inner scope must not interfere with correctness or lifetime analysis of enclosing scopes.
Consider nested scopes of the form
\[
S_{\mathrm{outer}} \;\{\, \dots \; S_{\mathrm{inner}} \;\{ \dots \}\; \dots \}.
\]
Let $v$ be introduced in $S_{\mathrm{inner}}$ and $w$ in $S_{\mathrm{outer}}$.

\begin{lemma}[Nested Scope Independence]
If $p$ is a safe reclamation point for $v$ inside $S_{\mathrm{inner}}$, inserting the adjoint segment associated with $v$ at $p$ does not alter the semantic lifetime of any variable $w$ defined outside $S_{\mathrm{inner}}$.
\end{lemma}

\begin{proof}
At a safe reclamation point, $v$ is non-live and disentangled from variables that remain live. In particular, it cannot influence outer-scope variables beyond $p$. The adjoint segment therefore affects only qubits local to $v$ and does not modify the dependence or entanglement relations among outer variables. Hence the semantic lifetime of $w$ is unchanged. \qed
\end{proof}

Thus early reclamation is stable under scope nesting.

\begin{lemma}[Lifetime Monotonicity]
Let $P$ be a program and $B$ a sub-block of $P$. For any variable $v$ defined in $B$, the semantic lifetime computed within $B$ is contained in its semantic lifetime computed in $P$.
\end{lemma}

\begin{proof}
Liveness is defined in terms of future dependencies and entanglement propagation. Restricting analysis to $B$ can only remove potential future uses, not introduce new ones. Therefore the maximal live point in $B$ does not exceed that in $P$. \qed
\end{proof}

Lifetime monotonicity ensures that reclamation decisions taken locally remain valid when the block is embedded into a larger program.

\begin{lemma}[Compositional Restoration]
If, for every lexical scope $S$, all temporaries introduced in $S$ satisfy the Restoration Invariant at the end of their semantic lifetime, then the entire program satisfies the Restoration Invariant.
\end{lemma}

\begin{proof}
Proceed by structural induction on scope nesting. The base case holds by assumption. For the inductive step, Nested Scope Independence ensures that reclamation inside inner scopes does not affect outer variables. Since each scope restores its own temporaries at lifetime boundaries and no residual entanglement escapes, the invariant holds at each enclosing level. \qed
\end{proof}

Compositionality extends to circuit-level metrics. Because reclamation introduces no new cross-scope dependencies, critical-path length and live-set bounds remain stable under structured composition. Early reclamation is therefore not a global optimization phase but a compositional semantic discipline, compatible with modular compilation and separate reasoning about program fragments.


\section{Semantic Consequences}

The formal model of Scope-Bounded Liveness-Guided Uncomputation has consequences beyond circuit-level depth and width. Since lifetime boundaries determine when temporary data must be restored, they also induce higher-level language semantics. In particular, parameter passing behavior and the conditions under which automatic uncomputation remains sound arise directly from semantic lifetime analysis and the Restoration Invariant.

\subsection{Parameter Passing Semantics}

In Qutes, parameter passing is not realized through distinct runtime mechanisms but emerges from lifetime boundaries. The observable effect of a function on its parameters depends on whether their semantic lifetime extends beyond the function body.

By default, Qutes enforces pass-by-value semantics. Consider
\begin{center}
\begin{tabular}{l}
\hline
\texttt{void f(qubit x)\{ ... \}} \\
\hline
\end{tabular}
\end{center}
Modifications applied to $x$ inside the body are provisional unless explicitly declared persistent. Since the semantic lifetime of $x$ as a parameter ends at function exit, the Restoration Invariant requires that its reduced state coincide with its initial state at the call boundary.

\begin{lemma}[Emergent Pass-by-Value]
Let $f$ be a function whose parameters are not declared persistent and whose body satisfies the Restoration Invariant with respect to those parameters. For any call $f(a)$, the observable quantum state of $a$ after the call equals its state before the call.
\end{lemma}

\begin{proof}
Within $f$, transformations on $x$ must either be propagated through returned values or be undone before the end of its semantic lifetime. Since parameters are not live beyond function exit unless explicitly returned or marked persistent, the Restoration Invariant ensures restoration at the boundary. \qed
\end{proof}

This behavior does not rely on copying quantum data, which would violate no-cloning, but on reversibility: internal modifications must vanish at the lifetime boundary.

Explicit pass-by-reference semantics is obtained by extending lifetime. With a \texttt{ref} annotation,
\begin{center}
\begin{tabular}{l}
\hline
\texttt{void g(ref qubit x)\{ ... \}} \\
\hline
\end{tabular}
\end{center}
the parameter remains live beyond the function boundary and is not subject to restoration at exit. Consequently, modifications persist and the observable effect matches pass-by-reference.

The distinction is purely lifetime-based: in pass-by-value mode, a parameter’s semantic lifetime ends at the function boundary; in pass-by-reference mode, it extends into the caller’s context.

For instance, the running example may be written as
\[
\texttt{void compute(}x_1,x_2,\texttt{ ref } y_1,\texttt{ ref } y_2\texttt{)},
\]
where $x_1,x_2$ are passed by value and $y_1,y_2$ by reference.

Thus parameter passing in Qutes is not a separate execution model but a direct manifestation of lifetime-guided uncomputation: persistence across boundaries must be explicit, and restoration is enforced otherwise.

\subsection{Semantic Constraints}

Lifetime-guided uncomputation is bounded by semantic and physical constraints arising from the Restoration Invariant and the principles of quantum computation. These constraints characterize precisely when automatic reclamation would violate reversibility or alter observable behavior.

A first constraint concerns irreversible operations. The formal model assumes unitary evolution, whereas measurement is inherently non-unitary and transfers information to the classical domain.

\begin{lemma}[Irreversibility Constraint]
Let $v$ be measured at program point $p$, and suppose the measurement outcome influences subsequent computation. Then $v$ cannot satisfy the Restoration Invariant beyond $p$.
\end{lemma}

\begin{proof}
Measurement induces a non-unitary transformation and extracts classical information. Since the post-measurement state depends on probabilistic outcomes that are no longer coherently encoded, no unitary segment can restore $v$ while preserving observable effects. \qed
\end{proof}

A second constraint concerns intentionally persistent entanglement. If correlations between input and output variables constitute the intended result of the computation, enforcing restoration of the input would destroy the desired joint state and alter semantics.

Aliasing introduces a further restriction. Liveness must be computed at the level of physical resources rather than syntactic identifiers.

\begin{lemma}[Aliasing Constraint]
Let $v$ and $w$ reference the same quantum register $R$. If one alias becomes non-live while another remains live, $R$ must be considered live. Reclamation is permitted only when $R$ itself is no longer live.
\end{lemma}

\begin{proof}
Since $v$ and $w$ denote the same physical resource, restoring one requires restoring $R$. If an alias remains live, its state must persist. Hence reclamation is allowed only when no alias of $R$ is live. \qed
\end{proof}

These constraints are unified by semantic lifetime: measurement introduces irreversible effects; persistent entanglement extends lifetime; aliasing prevents premature reclamation at the identifier level. Automatic uncomputation is therefore restricted to variables whose evolution remains unitary and whose lifetime genuinely terminates.

A final constraint concerns entanglement across lifetime boundaries. A temporary variable $x$ may interact with a variable $y$ whose lifetime extends beyond a boundary $B$. Restoration at $B$ is sound only if the induced correlations can be semantically isolated.

\begin{lemma}[Cross-Boundary Entanglement Constraint]
Let $x$ have semantic lifetime ending at boundary $B$, and let $y$ remain live beyond $B$. If prior operations entangle $x$ and $y$, restoration of $x$ at $B$ is possible if and only if the subcircuit involving $x$ is output-isolable with respect to $y$.
\end{lemma}

\begin{proof}
If output-isolability holds, the subcircuit admits a decomposition $U = U_y \circ V_x$, where $V_x$ acts only on $x$ and can be inverted without modifying the reduced state of $y$. Inserting $V_x^\dagger$ restores $x$ while preserving the observable state of $y$. Conversely, if no such decomposition exists, any inverse restoring $x$ necessarily modifies the joint state and alters the reduced density matrix of $y$, violating the Restoration Invariant. \qed
\end{proof}

Thus entanglement alone does not preclude reclamation; only correlations that cannot be isolated at the lifetime boundary prevent restoration. When output-isolability fails, automatic uncomputation cannot be guaranteed and restoration must be handled explicitly. Lifetime-guided uncomputation is therefore precisely delimited by the formal boundaries of reversible reasoning.


\section{Conclusions}

We introduced Scope-Bounded Liveness-Guided Uncomputation as a semantic discipline for quantum program compilation in Qutes. Instead of appending a global cleanup phase, uncomputation is governed by semantic lifetime and a Restoration Invariant that determines when temporary quantum data must be restored. Lifetime is defined at the level of physical resources, ensuring sound handling of aliasing.

Early reclamation is locally decidable and compositional under nested scopes. At the circuit level, it preserves worst-case asymptotics while enabling structural improvements: adjoint segments may be scheduled off the critical path, reducing depth, and temporary registers may be reclaimed early, reducing peak live width. Resource usage is thus determined by the live set rather than lexical allocation.

We characterized the semantic boundary of automatic reclamation through output-isolability: entanglement alone does not prevent restoration, but correlations that cannot be semantically isolated do. When isolability cannot be certified, restoration guarantees no longer apply.
Parameter passing semantics follows from the same lifetime discipline, with pass-by-value and pass-by-reference corresponding to distinct restoration policies.

Uncomputation in Qutes is therefore treated as a structural semantic principle integrating lifetime analysis, entanglement structure, and reversible resource management. Future work includes refined analyses for isolability and interprocedural lifetime verification.

\bibliographystyle{plain} 
\bibliography{bibliography.bib}

@inproceedings{grover,
  title={A fast quantum mechanical algorithm for database search},
  author={Grover, Lov K},
  booktitle={Proceedings of the twenty-eighth annual ACM Symposium on Theory of Computing},
  pages={212--219},
  year={1996},
}

@inproceedings{Quipper,
  title={Quipper: a scalable quantum programming language},
  author={Green, Alexander S and Lumsdaine, Peter LeFanu and Ross, Neil J and Selinger, Peter and Valiron, Beno{\^\i}t},
  booktitle={Proceedings of the 34th ACM SIGPLAN Conference on Programming Language Design and Implementation},
  year={2013},
}

@inproceedings{silq,
  title={Silq: A High-Level Quantum Language with Safe Uncomputation and Intuitive Semantics},
  author={Bichsel, Benjamin and Zhan, Daniel and Sutter, Goran and Vechev, Martin},
  booktitle={Proceedings of the 41st ACM SIGPLAN Conference on Programming Language Design and Implementation (PLDI)},
  doi={10.1145/3385412.3386007},
  year={2020},
}

@inproceedings{qsharp,
  title={{Q\#: Enabling Scalable Quantum Computing and Development with a High-level DSL}},
  booktitle={Proceedings of the Real World Domain Specific Languages Workshop 2018},
  author={Svore, Krysta and Geller, Alan and Troyer, Matthias and Azariah, John and Granade, Christopher and Heim, Bettina and Kliuchnikov, Vadym and Mykhailova, Mariia and Paz, Andres and Roetteler, Martin},
  doi={10.1145/3183895.3183901},
  year={2018},
}

@inproceedings{FaroMM25,
  author       = {Simone Faro and
                  Francesco Pio Marino and
                  Gabriele Messina},
  title        = {Qutes: {A} High-Level Quantum Programming Language for Simplified
                  Quantum Computing},
  booktitle    = {Proceedings of the 34th International Symposium on High-Performance
                  Parallel and Distributed Computing, {HPDC} 2025, University of Notre
                  Dame, Notre Dame, IN, USA},
  pages        = {47:1--47:9},
  publisher    = {{ACM}},
  year         = {2025},
  url          = {https://doi.org/10.1145/3731545.3744153},
  doi          = {10.1145/3731545.3744153},
  bibsource    = {dblp computer science bibliography, https://dblp.org}
}

@article{Qutes_Journal,
    author = {Faro, Simone and Marino, Francesco Pio and Messina, Gabriele},
    title = {Extending Qutes: a practical high-level language for quantum computing},
    journal = {The Computer Journal},
    pages = {bxaf133},
    year = {2025},
    month = {11},
    issn = {0010-4620},
    doi = {10.1093/comjnl/bxaf133},
    url = {https://doi.org/10.1093/comjnl/bxaf133},
    eprint = {https://academic.oup.com/comjnl/advance-article-pdf/doi/10.1093/comjnl/bxaf133/65613716/bxaf133.pdf},
}

@book{NielsenChuang,
  author    = {Michael A. Nielsen and Isaac L. Chuang},
  title     = {Quantum Computation and Quantum Information},
  publisher = {Cambridge University Press},
  edition   = {10th Anniversary Edition},
  year      = {2010},
  isbn      = {9781107002173},
  doi       = {10.1017/CBO9780511976667}
}

@article{landauer1961,
  author  = {Landauer, Rolf},
  title   = {Irreversibility and Heat Generation in the Computing Process},
  journal = {IBM Journal of Research and Development},
  volume  = {5},
  number  = {3},
  pages   = {183--191},
  year    = {1961},
  doi     = {10.1147/rd.53.0183}
}

@article{bennett1973,
  author  = {Bennett, Charles H.},
  title   = {Logical Reversibility of Computation},
  journal = {IBM Journal of Research and Development},
  volume  = {17},
  number  = {6},
  pages   = {525--532},
  year    = {1973},
  doi     = {10.1147/rd.176.0525}
}

@article{lecerf1963,
  author  = {Lecerf, Yves},
  title   = {Machines de Turing r{\'e}versibles. R{\'e}cursive insolubilit{\'e} en $n \in \mathbb{N}$ de l'{\'e}quation $u=\theta^n u$, o{\`u} $\theta$ est un ``isomorphisme de codes''},
  journal = {Comptes Rendus Hebdomadaires des S{\'e}ances de l'Acad{\'e}mie des Sciences},
  volume  = {257},
  pages   = {2597--2600},
  year    = {1963}
}

@techreport{toffoli1980,
  author      = {Toffoli, Tommaso},
  title       = {Reversible Computing},
  institution = {MIT Laboratory for Computer Science},
  type        = {Technical Memo},
  number      = {MIT-LCS-TM-151},
  year        = {1980},
  url         = {https://publications.csail.mit.edu/lcs/pubs/pdf/MIT-LCS-TM-151.pdf}
}

@inproceedings{unqomp,
  author    = {Paradis, Anouk and Bichsel, Benjamin and Steffen, Samuel and Vechev, Martin},
  title     = {Unqomp: Synthesizing Uncomputation in Quantum Circuits},
  booktitle = {Proceedings of the 42nd ACM SIGPLAN International Conference on Programming Language Design and Implementation (PLDI '21)},
  year      = {2021},
  pages     = {222--236},
  publisher = {Association for Computing Machinery},
  doi       = {10.1145/3453483.3454040}
}

@article{reqomp,
  author  = {Paradis, Anouk and Bichsel, Benjamin and Vechev, Martin},
  title   = {Reqomp: Space-constrained Uncomputation for Quantum Circuits},
  journal = {Quantum},
  volume  = {8},
  pages   = {1258},
  year    = {2024},
  month   = feb,
  doi     = {10.22331/q-2024-02-19-1258},
  url     = {https://doi.org/10.22331/q-2024-02-19-1258}
}

@misc{venev2024modular,
  author       = {Venev, Hristo and Gehr, Timon and Dimitrov, Dimitar and Vechev, Martin},
  title        = {Modular Synthesis of Efficient Quantum Uncomputation},
  howpublished = {arXiv preprint},
  eprint       = {2406.14227},
  primaryClass = {quant-ph},
  year         = {2024},
  url          = {https://arxiv.org/abs/2406.14227}
}

\end{document}